\documentclass[11pt,a4paper]{article}
\usepackage[utf8]{inputenc} 

\usepackage{authblk}

\usepackage{a4wide}
\usepackage{graphicx}
\usepackage{xspace}
\usepackage{amsthm}
\usepackage{amsfonts}
\usepackage{amsmath}
\usepackage{amssymb}
\usepackage{appendix}
\usepackage{color}
\usepackage{url}
\newcommand{\ilpeace}{\textsc{ilpeace}\xspace}
\newtheorem{theorem}{Theorem}

\newcounter{defcounter}
\newtheorem{definition}[defcounter]{Definition}

\newcommand{\R}{\mathbb{R}}

\newcommand{\cI}{\mathcal{I}}

\newcommand{\E}{\ensuremath{\mathbb{E}}\xspace}
\renewcommand{\S}{\ensuremath{\mathbb{S}}\xspace}
\newcommand{\D}{\ensuremath{\mathbb{D}}\xspace}
\newcommand{\T}{\ensuremath{\mathbb{T}}\xspace}
\renewcommand{\L}{\ensuremath{\mathbb{L}}\xspace}
\newcommand{\SL}{\ensuremath{\mathbb{SL}}\xspace}
\newcommand{\TL}{\ensuremath{\mathbb{TL}}\xspace}
\newcommand{\DTL}{\ensuremath{\mathbb{DTL}}\xspace}
\newcommand{\no}{\ensuremath{\varnothing}\xspace}
\newcommand{\cont}{\ensuremath{\mathbb{C}}\xspace}
\newcommand{\setminusC}{\text{\hspace{-0.8mm}}\setminus\text{\hspace{-0.8mm}}}
\allowdisplaybreaks
\usepackage{lipsum}
\newcounter{examplecounter}
\newenvironment{example}{\begin{quote}%
    \refstepcounter{examplecounter}%
  \textbf{Example \arabic{examplecounter}}%
  \quad
}{%
\end{quote}%
}
\newtheorem{lemma}[theorem]{Lemma}

\begin{document}




\title{Exact reconciliation of undated trees}

\author[1]{Leo van Iersel \thanks{l.j.j.v.iersel@gmail.com. Leo van Iersel was supported by a Veni grant of The Netherlands Organisation for Scientific Research (NWO)}}
\author[2]{Celine Scornavacca\thanks{celine.scornavacca@univ-montp2.fr}}
\author[3]{Steven Kelk\thanks{steven.kelk@maastrichtuniversity.nl}}
\affil[1]{Centrum Wiskunde \& Informatica (CWI), Amsterdam, The Netherlands}
\affil[2]{ISEM, CNRS -- Universit\'e Montpellier II, Montpellier, France}
\affil[3]{Department of Knowledge Engineering (DKE), Maastricht University, Maastricht, The Netherlands}

\maketitle

\begin{abstract} 
Reconciliation methods aim at recovering macro evolutionary events and at localizing them in the species history, by observing discrepancies between gene family trees and species trees.
In this article we introduce an Integer Linear Programming (ILP) approach for the NP-hard problem of
computing a most parsimonious time-consistent reconciliation of a gene tree {with} a species tree when dating information on speciations is not available. The ILP formulation, which builds upon the \DTL model, returns a most parsimonious reconciliation ranging over all possible datings of the nodes of the species tree. By studying its performance on plausible simulated data we conclude that the ILP approach is significantly faster than a brute force search through the space of all possible species tree datings. 
Although the ILP formulation is currently limited to small trees, we believe that it is an important proof-of-concept which opens the door to the possibility of developing an exact, parsimony based approach to dating species trees.
The software (\ilpeace) is freely available from:  \url{ http://homepages.cwi.nl/~iersel/ilpeace/}
\end{abstract}


\section{Background}
Reconciliation methods aim at recovering macro evolutionary events -- such as speciations, losses and duplications of genes -- and at locating them in the species history, by comparing gene family trees to species trees (for a review, see~\cite{Briefing2011}). These methods are often used when studying genome evolution as well as for inferring orthology relationships \cite{Storm2002orthology,Heijden2007orthology} and for improving phylogeny inference and dating \cite{Abby2010TransferAndDatation,Akerborg2009PhylogenyAndRec,Boussau2013PhylogenyAndRec}.

Here, we consider the problem of finding the Most Parsimonious Reconciliation (MPR) when considering -- as possible macro-events that shape the genome -- speciations, duplications, transfers and losses of genes. The general problem of finding a MPR is known to be NP-complete, even for reconciling two binary trees \cite{TOFIGH2011}. The complexity of the problem is due to the difficulty of ensuring the \emph{time-consistency} of gene transfers, i.e. handling the chronological constraints among nodes of the species tree that are induced by transfer events. However, the problem becomes polynomially solvable when accepting a dated species tree as input \cite[among others]{CONOW2010,doyon2011efficient,TOFIGHTHESIS}. In this paper, we extend the combinatorial reconciliation model introduced by Doyon \emph{et al.}~\cite{doyon2011efficient} (called the ``\DTL'' model), which can be used to solve this special case of the problem {(i.e. when the species tree is dated)}. Although relative dates -- obtainable by relaxed molecular clock techniques --  are sufficient to ensure tractability, this information is not available for all portions of the Tree of Life. The question, therefore, is how to deal with the NP-hardness of the general (i.e. undated) version of the problem.

We note here that solving the undated version of the problem without ensuring the time-consistency of gene transfers can be done in polynomial time \cite{hallettlagergrentofigh2004}.
Moreover, algorithms to solve the undated version of the problem are fast and often find temporally feasible solutions \cite[among others]{TOFIGH2011,Bansal2012}. However, what should be done if the optimal solutions returned by such algorithms are \emph{not} temporally feasible and/or if optimal solutions under the temporally feasible model are strictly less parsimonious than their infeasible counterparts? In this case algorithms that do not enforce temporal feasibility convey only limited information {(especially when the algorithm does not indicate if a produced solution is temporally feasible or not)}. An alternative algorithmic approach is to heuristically search through the space of temporally feasible solution. This is the approach taken by (amongst others) \cite{CONOW2010}. However, such heuristics offer no guarantees that they will locate the most parsimonious solution. Without such guarantees the MPR model is weakened dramatically, because no indication {is given on} how far we are from the most parsimonious solution. Ideally, therefore, we require an algorithm that is \emph{guaranteed} to compute the most parsimonious temporally feasible solution.
 
In this paper, we propose such an algorithm. In particular, we {present} a flexible Integer Linear Programming (ILP) formulation for finding an MPR  when  {some or all of the} dates are  unknown.  {Essentially, the ILP formulation computes
{an} MPR  ranging over all possible datings of the species tree. Given that
the ILP is built upon the model of Doyon \emph{et al.}, the time-consistency of the reconciliation
computed by the ILP is guaranteed. Although ILP has been used earlier in the reconciliation literature \cite{chang2011ilp,than2009species}, this is the first attempt to tackle the \DTL reconciliation model using this 
technique.}

We have embedded the ILP formulation in the software package \ilpeace and made
this publicly available \cite{ilpeace}. The software computes the MPR between a given binary gene tree
and a given undated (or partially dated) binary species tree and outputs a visualisation of
the optimal reconciliation in the form of a phylogenetic network and a description of all reconciliation events in a format compatible with the reconciliation editor SylvX \cite{SylvX}. To validate the method (both in terms of correctness and running time) we have compared its performance, {on plausible simulated data}, to an {algorithm} that simply brute forces over all possible datings of the species tree. Such a comparison is reasonable because with a mathematical model as complex as the \DTL model it is far from obvious how one can move beyond brute force i.e. how one can intelligently prune the search space. Indeed, this is a strong motivation for our use of ILP in the first place. {Our experiments show that} \ilpeace is typically 10-100 times faster than the brute force approach and, although it is still limited to relatively small trees, we believe that \ilpeace is nevertheless an important proof-of-concept, mirroring the emergence of proof-of-concept ILP formulations elsewhere in phylogenetics, e.g. \cite{ilpsupertree}. Enhancements to the ILP formulation are likely to open the door to the exciting possibility of, in the future, using a set of gene trees to impose a ``most parsimonious dating'' upon an undated species tree, similar to recently undertaken work in the maximum likelihood framework~\cite{Szoi23102012}.

\section{Methods}

\subsection{Basics}

The node set, edge set, internal node set and  leaf node set of a tree $T$ are respectively denoted $V(T)$, $E(T)$, $I(T)$ and $L(T)$.
Moreover, the label of each leaf $u$ is  denoted by $\mathcal{L}(u)$, while 
the  set of labels of leaves of $T$ is denoted by $\mathcal{L}(T)$. The root node of $T$ is denoted by $\rho(T)$.
Given two nodes $u$ and $v$ of  {a rooted tree} $T$, we write $u \leq_{T} v$ if and only if $v$ is on the unique path from $u$ to the root of $T$. If $u \leq_{T} v$ and $u\neq v$ then we write $u <_{T} v$.
For an internal node $u$ of $T$  {with two children}, let~$u_l, u_r$  {denote the two children (in arbitrary order)}. In this paper, we  {assume} that gene and species trees are rooted, binary and uniquely leaf-labeled  {i.e. within
each tree there is a bijection between leaves and labels. Due to this bijectivity we will often refer to leaves and labels
interchangeably}. The height of a node $u$ in a tree $T$ is denoted by $h_T(u)$ while {the} height of~$T$ is denoted by $h(T)$. 

We define a gene tree $G$ as a tree where each leaf represents an extant gene. Similarly, a species tree $S$ is defined as a tree in which each leaf represents a distinct extant species.  Note that several leaves of a gene tree can be associated to the same species due to duplication and transfer events. 
 {Formally, we indicate this by} a surjective function $s :  {\mathcal{L}}(G) \rightarrow \mathcal{L}(S)$,  {called the \emph{species labeling} of $G$}, see Fig. \ref{Fig:trees} for an example. The set of species labels of the leaves of $G$ is denoted $\mathcal{S}(G)$.

\begin{figure}[!ht]
\centering 
 \includegraphics[scale=0.35]{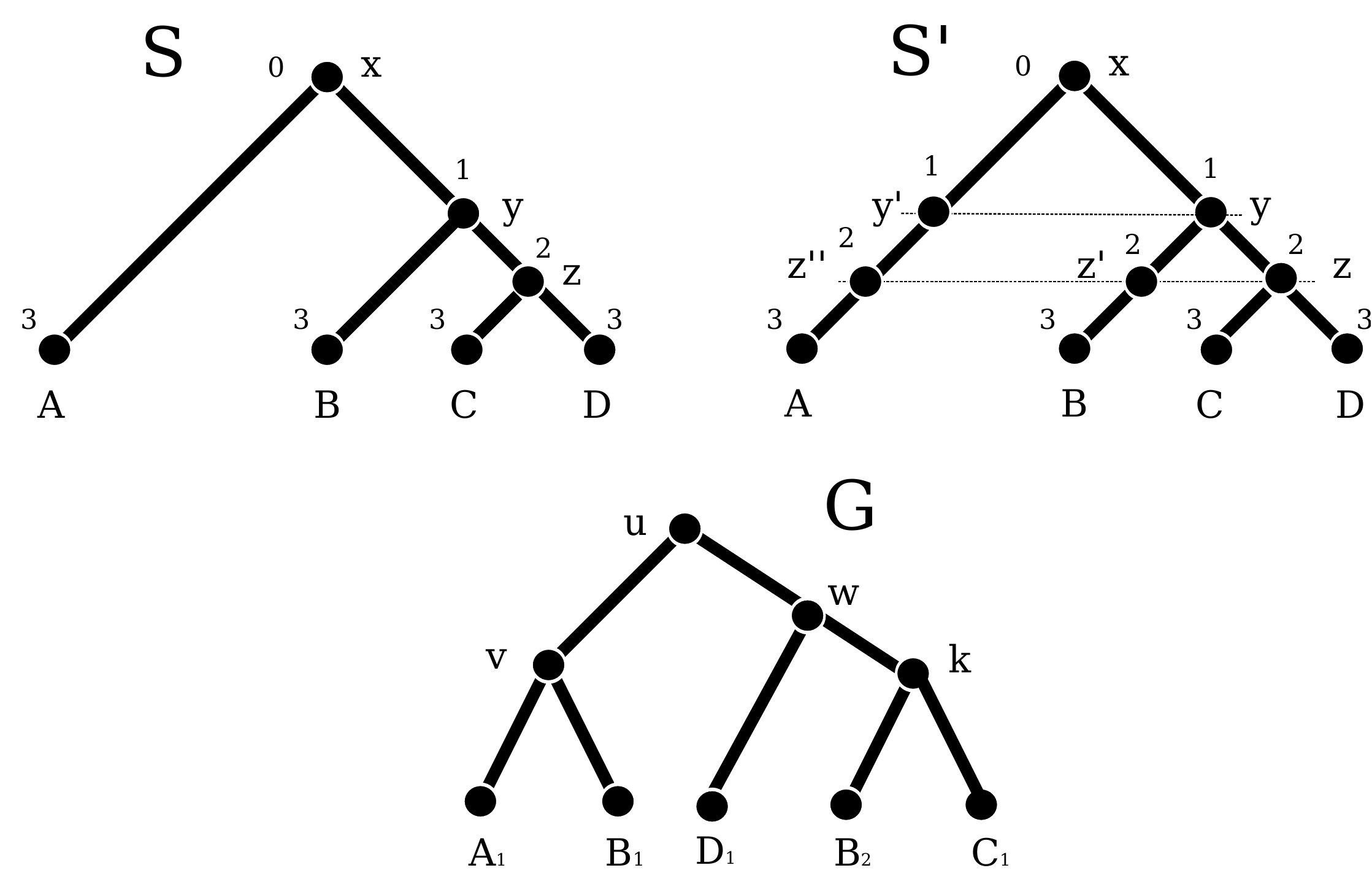}
\caption{ An example of a dated species tree ($S$) along with its  subdivision ($S'$), and of a gene tree ($G$).  The numbering of nodes of $S$ and $S'$ corresponds to dates. The species labeling of $G$ is the following: $s(A_1)=A$, $s(B_1)=s(B_2)=B$, $s(C_1)=C$ and $s(D_1)=D$.   \label{Fig:trees} }
\end{figure}

A species tree $S$ is said to be \emph{dated} if it is associated {with} a function  $\theta_{S}:V(S) \to \mathbb{R}^+$ such that {$\theta_S(\rho(S))=0$ {and} if $y {\geq_{S} } x$ then $\theta_S(y) \geq \theta_S(x)$
\footnote{ {Note that in the original definition of dating \cite{doyon2011efficient}, dates are assumed to decrease towards the leaves. This is only a superficial difference.}}.
As mentioned in the introduction, dates  {make it possible} to solve the MPR problem in polynomial time. The recently  proposed \DTL model  \cite{doyon2011efficient} uses  a \emph{subdivision} of $S$ rather than $S$ itself when computing reconciliations  to ensure time-consistency of gene transfers in polynomial time.  Given a tree $S$ and a time function $\theta_{S}$, the  subdivision $S'$ of $S$  {(together with a new time function $\theta_{S'}$)} is constructed  as follows: firstly, for each node $x \in {I(S)}$ and each edge $(y_p,y) \in E(S)$ s.t. $\theta_S(y_p) < \theta_S(x) < \theta_S(y)$, an \emph{artificial} node $w$ (i.e. a node with only one child) is inserted along the edge $(y_p,y)$, with $\theta_{S'}(w) =\theta_S(x)$; secondly, for nodes $x\in V(S')$ corresponding to nodes already  present in $S$, we set $\theta_{S'}(x) =\theta_S(x)$. 
An example of a  {dated} species tree along with its subdivision  {is shown in} Fig.~\ref{Fig:trees}.

The \DTL model  \cite{doyon2011efficient} reconciles a dated binary species tree $S$ with a binary gene tree $G$ by building a mapping $\alpha$ that {maps} each {node} $u \in V (G)$ to an ordered list of nodes in $V(S')$, namely the ancestral and/or extant species in which the sequence {corresponding to}~$u$ evolved.  
This model takes into account four kinds of biological events: speciations, duplications, transfers and losses of genes.  
The atomic events of this model are: a speciation (\S), a duplication (\D), a transfer (\T), a transfer followed immediately by the loss of the non-transferred child (\TL), a speciation followed by the loss of one of the two resulting children (\SL), and a contemporary event (\cont) that associates an extant gene to its corresponding species. Finally, a no event ($\no$), {is used to model} the fact that a gene lineage has crossed a time boundary. {Note that duplication-loss events,  unlike transfer-loss and speciation-loss events, 
 leave no trace and are therefore undetectable. This is why, in the \DTL model, losses are never considered alone. }
The formal definition of a \DTL reconciliation \cite{doyon2011efficient} is given below: 

\begin{definition}[\cite{doyon2011efficient}]
  \label{Def:D}
  Consider a gene tree $G$, a dated species tree $S$ such that $\mathcal{S}{(G)} \subseteq \mathcal{L}{(S)}$,
  %
  %
   and its subdivision $S'$.
  Let  $\alpha$ be a function that maps each node $u$ of $G$ onto
  an ordered sequence of nodes of $S'$, denoted
  $\alpha(u)=(\alpha_1(u), \alpha_2(u), \ldots, \alpha_\ell(u))$. 
  The function $\alpha$ is said to be a \emph{reconciliation} between $G$ and
  $S'$ if and only if exactly one of the 
  following events occurs for each couple of nodes $u$ of $G$
  and $\alpha_i(u)$ of $S'$ (denoting $\alpha_i(u)$ by $x'$ below):
  \begin{itemize}
  \item [a)] if $x'$ is the last node of $\alpha(u)$, one of the cases below is true:
    \begin{enumerate}
    \item[1.]  $u \in \mathcal{L}(G)$, $x' \in L(S')$ and  $s(\mathcal{L}(u))={\mathcal{L}}(x')$;  \hfill $($\cont event$)$
    \item[2.] $\{\alpha_1(u_l), \alpha_1(u_r)\} =\{x'_l,x'_r\}$;\hfill $($\S event$)$
    \item[3.] $\alpha_1(u_l)=x'$ and $\alpha_1(u_r)=x'$; \hfill $($\D event$)$
    \item[4.] $\alpha_1(u_l)=x'$, and $\alpha_1(u_r)$ is any node other than $x'$ having date $\theta_{S'}(x)$ 
    \item[or] $\alpha_1(u_r)=x'$, and $\alpha_1(u_l)$ is any node other than $x'$ having date $\theta_{S'}(x')$;
    \item[] 
    \hfill $($\T event$)$
   \end{enumerate}
    \newcounter{enum}
  \item  [b)] otherwise, one of the cases below is true:
   \begin{enumerate}
    \item[5.] $x'$ is an artificial node  and  $\alpha_{i+1}(u)$ is its only child;\hfill $($\no event$)$
    \item[6.] $x'$ is not artificial and $\alpha_{i+1}(u) \in
      \{x'_l,x'_r\}$; \hfill$($\SL event$)$ 
    \item[7.] $\alpha_{i+1}(u)$ is any node other than $x'$ having date $\theta_{S'}(x)$.
      \hfill   $($\TL event$)$
    \end{enumerate}
  \end{itemize}
\end{definition}

Note that, even if $\D$, $\T$ and $\TL$ events occur along a branch of $S'$, the mapping is done on nodes of $S'$. The cost of a reconciliation  $c(\alpha)$ is the sum $c_{\D}d +c _{\T}t + c_{\L}$l, where {$d$ is the number of~$\D$ events, $t$ is the number of~$\T$ and~$\TL$ events, and $l$ is the number of~$\SL$ and~$\TL$ events in $\alpha$, and $c_{\D}$, $c_{\T}$, $c_{\L}$ are, respectively, the costs of duplications, transfers and losses.}

\begin{example}\label{ex:recDef1}
An example of reconciliation is depicted in Fig.~\ref{Fig:rec}. This reconciliation corresponds to the following mapping $\alpha$: 
$\alpha(u)=\{x\}$ (event $\S_1$), 
$\alpha(v)=\{y'\}$ (event $\D_2$),
$\alpha(A_1)=\{y',z'',A\}$ (events $\no_3$, $\no_7$ plus {a} \cont event),
$\alpha(B_1)=\{y',z'',z',B\}$ (events $\no_4$, $\TL_6$ $\no_8$ plus {a} \cont event),
$\alpha(w)=\{y,z\}$ (events $\SL_5$, $\S_9$),
$\alpha(k)=\{C\}$ (event $\T_{10}$),
$\alpha(D_1)=\{D\}$ ({a} \cont event),
$\alpha(B_2)=\{B\}$ ({a} \cont event),
and $\alpha(C_1)=\{C\}$ ({a} \cont event).
\end{example}

\begin{figure}[!ht]
\centering
 \includegraphics[scale=0.5]{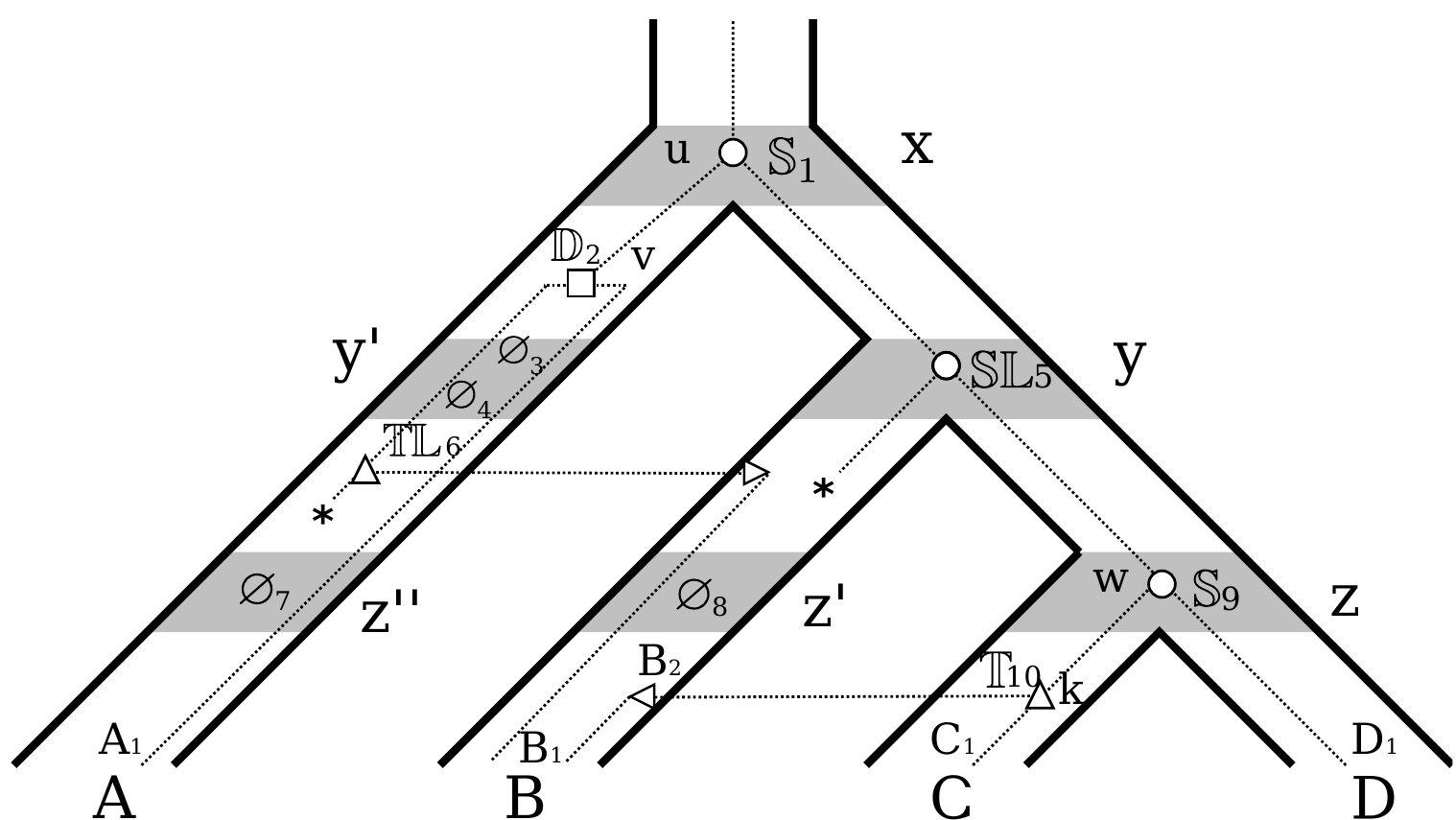}
 \caption{ An example of reconciliation for the trees depicted in Fig. \ref{Fig:trees}  containing two $\S$ events, one $\SL$ event, one \D event, a $\T$ event, a $\TL$ event, four \no events and five \cont events (\cont events are not indicated) \label{Fig:rec}.}
\end{figure}

\subsection{Dated Reconciliations}

Definition \ref{Def:D} assumes that the dates of all nodes are known. In particular, the subdivision~$S'$ of the species tree~$S$ in this definition is based on the known dates of the nodes. To deal with situations where the dates of the species tree are not known, we propose the following modified definition.  {We use $p(.)$ to refer to the parent of a node.}

\begin{definition}
  \label{Def:D2}
  Consider a gene tree $G$ and a species tree $S$ such that $\mathcal{S}{(G)} \subseteq \mathcal{L}{(S)}$.
    Let  $\alpha$ be a function that maps each node $u$ of $G$ onto
  an ordered sequence of nodes of $S$, denoted
  $\alpha(u)=(\alpha_1(u), \alpha_2(u), \ldots, \alpha_\ell(u))$. 
  Moreover, let ~$\tau$ be a function that assigns an ordered sequence of dates $\tau(u)=(\tau_1{(u)},\ldots ,\tau_{\ell}{(u)})\in\R^\ell$ to each node~${u}$ of~$G$, and let~$t$ be a function that assigns a date~$t(s)\in\R^+$ {to each node~$s$ of~$S$. The triple}~$(\alpha,\tau,t)$ is said to be a \emph{dated reconciliation} for~$(G,S)$  if and only if exactly one of the 
  following events occurs for each couple of nodes $u$ of $G$
  and $\alpha_i(u)$ of $S$ (denoting $\alpha_i(u)$ by $x$ below):

  \begin{itemize}
  \item [a)] if $x$ is the last node of $\alpha(u)$, one of the cases below is true:
    \begin{enumerate}
    \item[1.]  ${u} \in  {L}(G)$, $x \in L(S)$,
$\text{s}(\mathcal{L}(u))=\text{{$\mathcal{L}$}}(x)$ and
$\tau_\ell{(u)} = t(x)$;  \hfill $($\cont event$)$
    \item[2.] $\{\alpha_1(u_l), \alpha_1(u_r)\} =\{x_l,x_r\}$ and
$\tau_\ell{(u)}=t(x)$;\hfill $($\S event$)$
    \item[3.] $\alpha_1(u_l)=x$, $\alpha_1(u_r)=x$ and
$t(p(x)) {< } \tau_\ell{(u)} {<} t(x)$; \hfill $($\D event$)$
    \item[4.] $\alpha_1(u_l)=x$,   and $y:=\alpha_1(u_r)$ is any node other than $x$ such that  
$t(p(x)) {<} \tau_\ell{(u)} {<} t(x)$  and $t(p(y)) {<} \tau_\ell{(u)} {<} t(y)$,
    or the same with $l$ and $r$ interchanged;
    \hfill $($\T event$)$
   \end{enumerate}
  \item  [b)] otherwise, one of the cases below is true:
   \begin{enumerate}
    \item[5.] $\alpha_{i+1}(u) \in
    \{x_l,x_r\}$ and $\tau_i{(u)} = t(x)$;
\hfill$($\SL event$)$ 
\item[6.] $\alpha_{i+1}(u)$ is any node other than $x$ such that $t(p(x)) {<} \tau_i{(u)} < t(x)$ and
\item[] $t(p(\alpha_{i+1}(u))) {<} \tau_i{(u)} {<} t(\alpha_{i+1}(u))$.
      \hfill   $($\TL event$)$
    \end{enumerate}
\end{itemize}
In addition, 
    \begin{enumerate}
\item [7.] $\tau_\ell(p(u)) < \tau_1{(u)}$ and, for each~$i\in\{1,\ldots ,|\alpha(u)|-1\}$, $ \tau_i{(u)} < \tau_{i+1}{(u)}$;
\item[8.] $t(p(s)) < t(s)$, for each node~$s$ of~$S$;
\item[9.] $t(\rho(S))=0$ and~$t(s) = |I(S)|$ for each $s\in L(S)$.

\end{enumerate}
\end{definition}

{When the species tree is not required to be ultrametric (which we do require in this paper), restriction~9 can be omitted from the above definition.}

The cost of a dated reconciliation $c(\alpha,\tau,t)$ is defined similarly to what is done for a reconciliation.


The main differences between Definition~\ref{Def:D} and~\ref{Def:D2} are that in the latter definition we work with the original species tree~$S$ and not with its subdivision, and that a date~$t(s)$ is assigned to each node~$s$ of~$S$ and a date~$\tau_i({u})$ to each event. For~$\cont,\S$ and~$\SL$ events, the date of the event is equal to the date of the corresponding node of~$S$. For~$\D$ events, the date of the event has to be between the date of the corresponding node of~$S$ and the date of its parent. For~$\T$ and~$\TL$ events, the date of the event has to be in-between the date of the donor node of the transfer and the date of its parent, and the date of the event has to be in-between the date of the recipient node of the transfer and the date of its parent. This, {along with constraint~7,} forces the dated reconciliation to be time-consistent.


{Since} this definition does not use a subdivision of the species tree but works with the original species tree, {it} does not use $\no$-events. The core of the idea is that the reconciliation itself indicates the dates of nodes of the species tree and also the dates on which events happen.
 
 The following lemmas ensure that dated reconciliations can be used to find a minimum cost reconciliation over all possible datings of the species tree.
\begin{lemma}
Given a reconciliation $\alpha$ for a gene tree $G$ and a dated species tree $(S,\theta_S)$, there exists a dated reconciliation $(\alpha',\tau,t)$ for $G$ and $S$ with cost $c(\alpha)$.
\end{lemma}
\begin{proof}[Sketch of the proof.]
Let $(S',\theta_{S'})$ be  the dated subdivision of the dated species tree $(S,\theta_S)$. The mapping $\alpha'(u)$ for each node $u \in V(G)$ can be obtained from $\alpha(u)$ by removing all the \no events; then the dates of the events can be obtained by setting $ \tau_i{(u)}$ to $ \theta_{S'}(\alpha'_i(u))$ for each node $u \in V(G)$ and $1\leq i\leq |\alpha'(u)|$. To obtain $t$, we set $t(x)=\theta_S(x)$ for each node $x \in V(S)$. It is straightforward to see that $(\alpha',\tau,t)$ satisfies Definition \ref{Def:D2} and that it has the same number of {events of each type as}~$\alpha$.
\end{proof}

 \begin{lemma}
Given a dated reconciliation $(\alpha',\tau,t)$ for a gene tree $G$ and a species tree $S$, there exists a reconciliation $\alpha$  for $G$ and $S,t$ with cost $c(\alpha',\tau,t)$.
\end{lemma}
\begin{proof}[Sketch of the proof.]
Constraints 8 and 9 of Definition \ref{Def:D2} ensure that the pair $(S,t)$ is a dated tree. Let $(S',\theta_{S'})$ be  the dated subdivision of the dated species tree. The mapping $\alpha(u)$ for each node $u \in V(G)$ can be obtained from $\alpha'(u)$ in the following way: we add before $\alpha'_1(u)$ all the descendants $x'$ of $p(\alpha'_1(u))$ in $S'$ having $t_{|\alpha'(u_p)|}(u_p) < \theta_{S'}(x') < t_{1}(u) $.
Moreover, for each $1\leq i < |\alpha'_i(u)|$, we add between $\alpha'_i(u)$ and $\alpha'_{i+1}(u)$ all the descendants $x'$ of $p(\alpha'_i(u))$ in $S'$ having $t_{i}(u) < \theta_{S'}(x') < t_{i+1}(u) $.

  It can be easily proved that $\alpha$ satisfies Definition \ref{Def:D} and that it has the same number of {events of each type as}~$(\alpha',\tau,t)$.
\end{proof}

 \begin{example}\label{ex:recDef2}
 
 Consider the reconciliation from Fig.~\ref{Fig:rec}. Since we do not use a subdivision, the reconciliation is given by the following mapping:

\smallskip
\begin{tabular}{ll}
$\alpha(u)=\{x\}$ ($\S$ event) & $\alpha(v)=\{A\}$ ($\D$ event)\\
$\alpha(A_1)=\{A\}$ (\cont event) & $\alpha(B_1)=\{A,B\}$ ($\TL$ and \cont event)\\
$\alpha(w)=\{y,z\}$ ($\SL$ and $\S$ event) & $\alpha(k)=\{C\}$ ($\T$ event)\\
$\alpha(D_1)=\{D\}$, $\alpha(B_2)=\{B\}$, &\hspace{-8mm} and  $\alpha(C_1)=\{C\}$ (all \cont events)\\
\end{tabular}

\smallskip
 {One possible assignment for $\tau$ and $t$ is:}
 
 \smallskip
\begin{tabular}{llll}$t(x)=0$  & $t(y)=1$ & $t(z)=2$ &  \hspace{-28mm} $t(A)=(B)=t(C)=(D)=3$\\
$\tau(u)=\{0\}$ & $\tau(v)=\{0.5\}$ & $\tau(A_1)=\{3\}$ &  \hspace{-28mm} $\tau(B_1)=\{1.5,3\}$ \\
$\tau(w)=\{1,2\}$  & $\tau(k)=\{2.5\}$ & $\tau(D_1)=\tau(B_2)=\tau(C_1)=\{3\}$. & 
\end{tabular}

\end{example}

Based on this new definition we propose an ILP formulation for finding an optimal dated reconciliation. An ILP consists of three main building blocks: the variables {and parameters} (described in Section~\ref{sec:var}), the objective function and the constraints (both described in Section~\ref{sec:ilp}). Although we omit the proof, the ILP formulation is mathematically equivalent to computing {a} most parsimonious {dated} reconciliation. {(More details on why the two formulations are equivalent will be given in the next section)}. We call this new method \ilpeace.

\subsection{Variables {and parameters}}\label{sec:var}

 Let $\mathcal{I}:=\{1,\ldots ,\ell\}$,  with~$\ell=|I(S)|$;  note that this is
a safe upper bound on the number of nodes of the species tree to which a node of the gene tree can be mapped in a most parsimonious reconciliation {(because otherwise one could modify the reconciliation to obtain one with fewer~$\SL$ and/or fewer~$\TL$ events).}
Let $\E:=\{\cont,\S,\D,\T,\SL,\TL\}$. Then we have a binary variable~$\alpha_{u,x,i,e}$, for each~$u\in V(G),x\in V(S), i \in \mathcal{I},e\in \E$, and $\alpha_{u,x,i,e}=1$ precisely if node~$x$ is the $i$-th node of~$V(S)$ that node~$u$ is mapped to and this mapping corresponds to an event of type~$e$. In addition, there is a variable~$\tau_{u,i}\in\R^+$, for each~$u\in V(G)$ and~$ i \in \mathcal{I}$, representing the time of the $i$-th event mapping node~$u$ to some node {in} $V(S)$ (with $\tau_{u,i}=\tau_{u,i-1}$ if there is no such event). There is also a variable~$t_s\in\{0,1,\ldots ,T_{\max}\}$, for each~$x\in V(S)$, representing the time of (the speciation event indicated by) node~$x$, with~$T_{\max} := |I(S)|$. By assumption,~$t_\rho(S)=0$ and~$t_x=T_{\max}$ for each leaf~$x$ of the species tree.

{The only parameters are the costs of the events. For} each event~$e\in \E$, $c_e$ indicates {its associated cost}.

 \begin{example}\label{ex:ILP}
 
Consider the dated reconciliation given in Example \ref{ex:recDef2}. 
Writing {the mapping described by} this reconciliation as binary variables, we get:

\smallskip
\begin{tabular}{lll}
$\alpha_{u,x,1,\S} = 1$ & $\alpha_{v, {A},1,\D} = 1$ & $\alpha_{A_1,A,1,\cont} = 1$\\
$\alpha_{B_1,A,1,\TL} = 1$ & $\alpha_{B_1,B,2,\cont} = 1$ &\\
$\alpha_{w,y,1,\SL} = 1$ & $\alpha_{w,z,2,\S} = 1$ & $\alpha_{k,C,1,\T} = 1$\\
$\alpha_{D_1,D,1,\cont} = 1$ & $\alpha_{B_2,B,1,\cont} = 1$ & $\alpha_{C_1,C,1,\cont} = 1$\\
\end{tabular}

\smallskip
and all other $\alpha$-variables are~0. The timing variables are:

\smallskip
\begin{tabular}{llll}
$t_x=0$ & $t_y=1$ & $t_z=2$&\\
$t_A=3$ & $t_B=3$ & $t_C=3$ & $t_D=3$\\
$\tau_{u,1\leq i\leq3}=0$ & $\tau_{v,1\leq i\leq3}=0.5$ & &\\
$\tau_{A_1,1\leq i\leq3}=3$ & $\tau_{B_1,1}=1.5$ & $\tau_{B_1,2\leq i\leq3}=3$&\\
$\tau_{w,1}=1$ & $\tau_{w,2\leq i\leq3}=2$ & $\tau_{k,1\leq i\leq3}=2.5$&\\
$\tau_{D_1,1\leq i\leq3}=3$ & $\tau_{B_2,1\leq i\leq3}=3$ & $\tau_{C_1,1\leq i\leq3}=3.$&\\
\end{tabular}
 
\end{example}

\medskip

{We argue that any dated reconciliation can be written in terms of such variables. To see this, let~$(\alpha,\tau,t)$ be a dated reconciliation. Suppose that~$\alpha_i(u)=x$. Then, by Definition~\ref{Def:D2}, there is exactly one event~$e\in\E$ corresponding to~$u,x$ and~$i$. In that case, we set the binary variable~$\alpha_{u,x,i,e}$ to~1 and we set~$\tau_{u,i}:=\tau_i(u)$.  Now suppose that $\alpha_i(u)$ does not exist, i.e.~$i >|\alpha(u)|$. In that case, we set $\alpha_{u,x,i,e}$ to~0 and $\tau_{u,i}$ to~$\tau_{u,i-1}$. In all other cases, i.e. if $\alpha_i(u)\neq x$, we also set $\alpha_{u,x,i,e}$ to~0. Finally, the dating variables for the species tree are simply obtained by setting~$t_x = t(x)$ for each vertex~$x$ of~$S$.}

{The constraints in the next subsection will enforce that the values of the binary variables produced by an ILP solver correspond to a dated reconciliation.}

\subsection{ILP Formulation}\label{sec:ilp}

The objective is to minimize the total cost of all events used in the reconciliation. Hence, the objective function is:

\[
\min \quad  \sum_{u\in V(G)}\sum_{x\in V(S)}\sum_{i \in \mathcal{I}}\sum_{e\in \E} c_e \alpha_{u,x,i,e}.
\]

Now we introduce the constraints needed to have our mapping $\alpha$ satisfy Definition~\ref{Def:D2}. Constraints~[1]-[7] model some general properties of the mapping. Constraint~[1] enforces that each vertex~$u\in V(G)$ is mapped to at most one vertex~$x\in V(S)$ in at most one type of event~$e\in\E$, for each index~$i\in\cI$. Constraint~[2] ensures that each~$u$ is mapped somewhere (i.e. to some~$x\in V(S)$ for some~$e\in\E$) for index~$i=1$. Constraint~[3] makes sure that if~$u$ is mapped somewhere for index~$i+1$, then~$u$ is also mapped somewhere for index~$i$. Constraint~[4] states that if~$u$ is mapped somewhere for index~$i+1$, then for index~$i$ it cannot be mapped to an event of type~$\cont,\S,\D$ or~$\T$. Constraint~[5] enforces that if~$u$ is mapped somewhere for index~$i$ with an event of type~$\SL$ or~$\TL$, then it has to be mapped somewhere for index~$i+1$. Finally, Constraints~[6] and~[7] state that no events of type~$\SL$ and~$\TL$ are allowed for the last index~$i=\ell$.

\begin{flalign*}
[1]\quad& \sum_{x\in V(S)}\sum_{e\in \E} \alpha_{u,x,i,e} \leq 1						&&\forall\text{ } u\in  {V(G)}, i \in \mathcal{I}\\
[2]\quad&\sum_{x\in V(S)}\sum_{e\in \E} \alpha_{u,x,1,e} = 1						&&\forall\text{ } u\in  V(G)\\
[3]\quad&\sum_{x\in V(S)}\sum_{e\in \E} (\alpha_{u,x,i,e} - \alpha_{u,x,i+1,e}) \geq 0		&&\forall\text{ } u\in  V(G), i\in \mathcal{I}\setminusC\{\ell\}\\
[4]\quad&\sum_{x\in V(S)}\sum_{e\in \E} \alpha_{u,x,i+1,e} + \sum_{x\in V(S)}\sum_{e\in \{\cont,\S,\D,\T\}} \alpha_{u,x,i,e} \leq 1\hspace{-1.3mm}		&&\forall\text{} u\in  V(G), i\in \mathcal{I}\setminusC\{\ell\}\\
[5]\quad& \sum_{x\in V(S)}\sum_{e\in \{\SL,\TL\}} \alpha_{u,x,i,e} - \sum_{x\in V(S)}\sum_{e\in \E} \alpha_{u,x,i+1,e} \leq 0		&&\forall\text{ } u\hspace{-0.8mm}\in \hspace{-0.8mm} V(G), i\in \mathcal{I}\setminusC\{\ell\}\\
[6]\quad&\text{ } \alpha_{u,x,\ell,\SL} = 0											&&\forall\text{ } u\in  V(G),x \in  V(S)\\
[7]\quad&\text{ } \alpha_{u,x,\ell,\TL} = 0											&&\forall\text{ } u\in  V(G), x \in  V(S)
\end{flalign*}

The next set of constraints models the \cont events. 
A \cont event maps a leaf of~$G$ to the leaf of $S$ given by the $s(.)$ function. 

%
{
\begin{flalign*}
[8]\quad& \displaystyle\sum_{i \in \mathcal{I}} \alpha_{u,x,i,\cont} = 1									&& \forall\text{ } u\in  L(G),x \in  L(S) \text{ with } s(\mathcal{L}(u))={\mathcal{L}}(x)\\
[9]\quad& \displaystyle\sum_{i \in \mathcal{I}} \alpha_{u,x,i,\cont} = 0									&&  \text{otherwise}
\end{flalign*}
}

We now model the \S events. If~$u$ is mapped to~$x$ in a \S event, then the children of~$u$ must be mapped to the children of~$x$. This can be enforced by the following constraints. Constraint~[10] states that leaves can not be mapped in \S events. Constraint~[11] enforces that if~$u$ is mapped somewhere in an event of type~$\S$, then its ``left'' child~$u_l$ must be mapped to one of the children~$x_l,x_r$ of~$x$ for index~$i=1$. Constraint~[12] does the same for the right child~$u_r$ of~$u$. Constraint~[13] then enforces that~$u_l$ and~$u_r$ cannot both be mapped to~$x_l$ and Constraint~[14] enforces that~$u_l$ and~$u_r$ cannot both be mapped to~$x_r$.

{
\begin{flalign*}
[10]\quad& \alpha_{u,x,i,\S} = 0										& \hspace{-1cm} \forall  i \in \mathcal{I},  \text{if } u&\in  L(G) \text{ or }x \in  L(S) \\
[11]\quad&   \sum_{i \in \mathcal{I}} \alpha_{u,x,i,\S} - \sum_{e\in \E} (\alpha_{{u}_l,x_l,1,e} + \alpha_{{u}_l,x_r,1,e}) \leq 0 					& 	& \forall\text{ } u\in  I(G),x \in  I(S)\\
[12]\quad&   \sum_{i \in \mathcal{I}} \alpha_{u,x,i,\S} - \sum_{e\in \E} (\alpha_{{u}_r,x_l,1,e} + \alpha_{{u}_r,x_r,1,e}) \leq 0						& & \forall\text{ } u\in  I(G),x \in  I(S)\\
[13]\quad&   \sum_{i \in \mathcal{I}} \alpha_{u,x,i,\S} + \sum_{e\in \E} (\alpha_{{u}_l,x_l,1,e} + \alpha_{{u}_r,x_l,1,e}) \leq 2						& &\forall\text{ } u\in  I(G),x \in  I(S)\\
[14]\quad&   \sum_{i \in \mathcal{I}} \alpha_{u,x,i,\S} + \sum_{e\in \E} (\alpha_{{u}_l,x_r,1,e} + \alpha_{{u}_r,x_r,1,e}) \leq 2						& &  \forall\text{ } u\in  I(G),x \in  I(S)
\end{flalign*}
}

This brings us to modeling  \D events. If~$u$ is mapped to~${x}$ in a \D event, then the children of~$u$ must also be mapped to~${x}$. This is modelled by the following two constraints. Constraint~[15] makes sure that leaves are not mapped in events of type~$\D$. Constraint~[16] enforces that if~$u$ is mapped to~$x$ in an event of type~$\D$, then both children~$u_l$ and~$u_r$ of~$u$ have to be mapped to~$x$.
 
 {
\begin{flalign*}
[15]\quad& \alpha_{u,x,i,\D} = 0												&&\forall\text{ } u\in  L(G), x \in  V(S), i \in \mathcal{I}\\
[16]\quad& \sum_{i \in \mathcal{I}} 2 \alpha_{u,x,i,\D} - \sum_{e\in \E} (\alpha_{{u}_l,x,1,e} + \alpha_{{u}_r,x,1,e}) \leq 0						&&\forall\text{ } u\in  I(G),x \in  V(S)
\end{flalign*}
}

Constraints [17]-[28] model general properties of the timing variables. First we need some definitions.  For a node~$v$ that is not the root, let~$p(v)$ denote its parent. Let $l_{p(v),v}$ denote the length of the edge between~$p(v)$ and~$v$ {(representing the elapsed time {between the two speciations})}, if known, and let $l_{p(v),v}=0$ indicate that the length of this edge is not known.  {Let~$\epsilon := 1 / \ell$.}

The following five constraints model the timing variables of the vertices~$V(S)$ of the species tree. Constraint~[17] states that the date of the root is~0. We assume that all leaves of~$S$ have date~$T_{\max}:=|I(S)|$ (Constraint~[18]) while the internal vertices of~$S$ have dates in~$\{0,1,\ldots ,T_{\max}-1\}$ (Constraint~[19]). (Note that these constraints could easily be relaxed for non-ultrametric species trees.) Constraint~[20] makes sure that the difference between the date of~$x$ and the date of its parent~$p(x)$ is~$l_{p(x),x}$ if this edge length is known. Constraint~[21] states that this difference should be at least~1 when the edge length is not known. This constraint is valid because without loss of generality the dates of internal vertices of~$S$ are all different elements of $\{0,1,\ldots ,T_{\max}-1\}$.

\begin{flalign*}
[17]\quad& t_{\rho(S)} = 0												&&\\
[18]\quad& t_{x} = T_{\max}												&&\forall\text{ } x\in L(S)\\
[19]\quad& t_{x} \leq T_{\max} - 1												&&\forall\text{ }x\in  I(S)\\
[20]\quad& t_x - t_{p(x)} = l_{p(x),x}												&&\forall\text{ }x\in  V(S)\setminusC\{\rho(S)\} \text{ with } l_{p(x),x} \neq 0\\
[21]\quad& t_x - t_{p(x)} \geq 1												&&\forall\text{ }x   \in  V(S)\setminusC\{\rho(S)\} \text{ with } l_{p(x),x} = 0
\end{flalign*}

The next three constraints model the timing variables for the mapping events. Constraint~[22] states that the date of the event mapping vertex~$u$ for index~$i=1$ should be strictly larger than the date of the last event mapping the parent~$p(u)$ of~$u$. Constraints~[23] and~[24] ensure that $\tau_{u,i+1} \geq \tau_{u,i}$ and that equality holds precisely if~$u$ is not mapped anywhere for index~$i+1$.

\begin{flalign*}
[22]\quad& {\tau_{u,1} - \tau_{p(u),\ell} \geq \epsilon} &&\forall\text{ } u\in  V(G)\\
[23]\quad& \tau_{u,i+1} - \tau_{u,i} - \sum_{x\in  V(S)} \sum_{e\in \E} \alpha_{u,x,i+1,e} \geq 0												             &&\forall\text{ } u\in  V(G), i\in \mathcal{I}\setminusC\{\ell\}\\
[24]\quad& \tau_{u,i+1} - \tau_{u,i} - T_{\max} \sum_{x\in  V(S)} \sum_{e\in \E} \alpha_{u,x,i+1,e} \leq 0												&&\forall\text{ } u\in  V(G), i\in \mathcal{I}\setminusC\{\ell\}
\end{flalign*}

Constraints~[25] and~[26] below make sure that the time of an event is in-between the time of the corresponding node~$x$ of the species tree and the time of its parent. If the event is of type \cont, \S or \SL, then the time of the event must be equal to the time of~$x$ (Constraint~[27]), and if the event is of type \D, \T or \TL, then the time of the event must be strictly smaller than the time of~$x$  (Constraint~[28]), {since the latter events happen on edges rather than nodes of the species tree. Note that we do not require here that $\tau_{u,i}$ is strictly greater than $t_{p(x)}$ because this is already implied by Constraint~[22]}. {Let~$V^+(S):=V(S)\setminusC\{\rho(S)\}$.}

\begin{flalign*}
[25]\quad& \tau_{u,i} - t_x + T_{\max} \sum_{e\in \E} \alpha_{u,x,i,e} \leq T_{\max} && \forall\text{ } u\in  V(G),x \in  V(S), i\in \mathcal{I}\\
[26]\quad& t_{p(x)} - \tau_{u,i} + T_{\max} \sum_{e\in \E} \alpha_{u,x,i,e} \leq T_{\max} && \forall\text{ } u\in  V(G), x \in  {V^+(S)}, i\in \mathcal{I}\\
[27]\quad& t_x - \tau_{u,i} + T_{\max} \sum_{e\in \{\cont,\S,\SL\}} \alpha_{u,x,i,e} \leq T_{\max} && \forall\text{ } u\in  V(G),x \in  {V^+(S)}, i\in \mathcal{I}\\
[28]\quad& \tau_{u,i} - t_x + T_{\max} \sum_{e\in \{\D,\T,\TL\}} \alpha_{u,x,i,e} \leq T_{\max} - \epsilon && \forall\text{ } u\in  V(G),x \in  {V^+(S)}, i\in \mathcal{I}
\end{flalign*}

We can now model \T events. First, Constraint~[29] states that leaves cannot be mapped in \T events. Constraint~[30] enforces that if~$u$ is mapped to~$x$ in a \T event, then one child of~$u$ must also be mapped to~$x$. Say that the other child of~$u$ is mapped to~$y$, then we must have~$y\in U(x)$, with~$U(x)$ the set of nodes that are not a descendant and not an ancestor of~$x$. This is ensured by Constraints~[31] and~[32].

\begin{flalign*}
[29]\quad&  \alpha_{u,x,i,\T} = 0			&	\hspace{-1cm}  \forall\text{ } u\in L(G), x\in  V(S), i \in \mathcal{I}&\\
[30]\quad&  \sum_{i \in \mathcal{I}} \alpha_{u,x,i,\T} - \sum_{e\in \E} (\alpha_{{u}_l,x,1,e} + \alpha_{{u}_r,x,1,e}) \leq 0& \forall\text{ } u\in  I(G), x\in V(S) &\\
[31]\quad& \sum_{i \in \mathcal{I}} \alpha_{u,x,i,\T} + \sum_{e\in \E} \alpha_{{u}_l,x,1,e} - \sum_{e\in \E} \sum_{y\in U(x)} \alpha_{{u}_r,y,1,e} \leq 1	&  \forall\text{ } u\in  I(G), x\in V(S) &\\
[32]\quad& \sum_{i \in \mathcal{I}} \alpha_{u,x,i,\T} + \sum_{e\in \E} \alpha_{{u}_r,x,1,e} - \sum_{e\in \E} \sum_{y\in U(x)} \alpha_{{u}_l,y,1,e} \leq 1	 & \forall\text{ } u\in  I(G), x\in V(S) &
\end{flalign*}

Constraints~[33] and~[34] enforce that if~$c\in\{u_l,u_r\}$ is mapped to~$y$ for index~$i=1$, then {$t_{p(y)} < \tau_{u,i}< t_y$}\footnote{{Actually, for keeping the ILP formulation as simple as possible, Constraints~[33] and~[34] enforce that  $t_{p(y)} \leq \tau_{u,\ell}\leq t_{y}$. If $\tau_{u,\ell}$ is equal to $t_{p(y)}$ (or to $t_{y}$), we can simply add (or subtract) $\epsilon/2$ to $\tau_{u,\ell}$ to satisfy the
strict inequality.}\label{footnote2}} . If~$u$ is mapped by an event of type~$\T$ for index~$i=\ell$, then this restriction is necessary to enforce time-consistency. Note that, if the last event mapping~$u$ is an event not of type~$\T$, then a child of~$u$ can only be mapped to either~$y$ or to a child of~$y$. In both cases, the restriction is valid. Since the restriction is always fulfilled for any event that is not a~$\T$,  for simplicity we impose it for all event types.

\begin{flalign*}
 [33]\quad&\tau_{u,\ell} - t_y + T_{\max} \sum_{e\in \E} \alpha_{c,y,1,e} \leq T_{\max} & \hspace{-1.5cm}  \forall\text{ } u\in  V(G), c\in\{u_l,u_r\},y\in  V(S) &\\
[34]\quad& t_{p(y)} - \tau_{u,\ell} + T_{\max} \sum_{e\in \E} \alpha_{c,y,1,e} \leq T_{\max} & \hspace{-1.5cm} \forall\text{ } u\in  V(G), c\in\{u_l,u_r\},y \in  {V^+(S)}
\end{flalign*}

The next two constraints model \SL events. Constraint~[35] states that leaves cannot be mapped in \SL events. Constraint~[36] enforces that if~$u$ is mapped to~$x$ for index~$i$ in an \SL event, then~$u$ must be mapped to one of the children of~$x$ for index $i+1$.

\begin{flalign*}
[35]\quad& \alpha_{u,x,i,\SL} = 0												&\forall\text{ } u\in  V(G),x \in  L(S), i\in \mathcal{I}\\
[36]\quad& \alpha_{u,x,i,\SL} - \sum_{e\in \E} ( \alpha_{u,x_l,i+1,e} + \alpha_{u,x_r,i+1,e} ) \leq 0						&\forall\text{ } u\in  V(G),x\in  I(S),  i \in \mathcal{I}\setminusC\{\ell\}\\
\end{flalign*}

We now formulate the constraints that model \TL events. If~$u$ is mapped to~$x$ for index~$i$ in a \TL-event, and if~$u$ is mapped to~$y$ for index~$i+1$, then we must have that $t_{p(y)} < \tau_{u,i}< t_{y}$ {(see footnote \ref{footnote2})}. This is enforced by Constraints~[37] and~[38]. Note that this holds automatically for \SL events, hence we do not need to restrict it to \TL events. Constraint~[39] enforces that~$y\notin U(x)$.

\begin{flalign*}
[37]\quad& \tau_{u,i} - t_y + T_{\max} \sum_{e\in \E} \alpha_{u,y,i+1,e} \leq T_{\max} && \forall\text{ } u\in  V(G),y\in  V(S), i\in \mathcal{I}\setminusC\{\ell\}\\
[38]\quad& t_{p(y)} - \tau_{u,i}  + T_{\max} \sum_{e\in \E} \alpha_{u,y,i+1,e} \leq T_{\max} && \forall\text{ } u\in  V(G),y \in {V^+(S)}, i\in \mathcal{I}\setminusC\{\ell\}\\
[39]\quad& \alpha_{u,x,i,\TL} - \sum_{y\in U(x)} \sum_{e\in \E} \alpha_{u,y,i+1,e}  \leq 0						&&\forall\text{ } u\in  V(G), x\in V(S), i\in \mathcal{I}\setminusC\{\ell\}
\end{flalign*}

Finally, we need the bounds and integrality constraints. Even though the time-variables could be restricted to be integer, this is not necessary.

\begin{flalign*}
[40]\quad& 0\leq t_x \leq T_{\max} &&\forall\text{ }x \in  V(S)\\
[41]\quad&  0\leq \tau_{u,i} \leq T_{\max} &&\forall\text{ } u\in  V(G), i \in \mathcal{I}\\
[42]\quad&  \alpha_{u,x,i,e}\in\{0,1\}			     								&&\forall\text{ } u\in  V(G),x \in  V(S), i \in \mathcal{I}, e\in \E
\end{flalign*}

This concludes the {ILP} formulation. 

For example, it can be easily checked that the given assignment of binary variables and dates in Example~\ref{ex:ILP} fulfills all the constraints of the ILP formulation for the trees in Figure \ref{Fig:trees} and thus is a valid {dated reconciliation} for these trees.

\section{Results and discussion}
\subsection{Implementation}

The \ilpeace method has been implemented in Java {and made publicly available~\cite{ilpeace}}. It generates an ILP formulation which is then solved by {the ILP solver} CPLEX \cite{ilog2013cplex}, {which is} a state-of-the-art ILP solver built upon a polyhedral \emph{branch, bound and cut} core \cite{mitchell2002branch}. {\ilpeace first computes an approximate solution} by setting the number of~\TL events to~0. This can be achieved by adding a constraint
\[
[43]\quad \sum_{u\in V(G)}\sum_{x\in V(S)}\sum_{i\in\mathcal{I}} \alpha_{u,x,i,\TL} = 0.
\]
Practical experiments show that an optimal solution to this restricted problem can usually be computed relatively quickly by CPLEX and that it provides a very good upper bound {(see Section~\ref{sec:results})}. {We call this the \emph{no-$\TL$-bound}.} The corresponding solution is then given to CPLEX as a ``warm start'' (i.e. an upper bound)  {from which CPLEX can search for the true optimum}. Once an optimal reconciliation has been found,
\ilpeace outputs the reconciliation score, number of events of each type, computation time, and the following files:  a file containing a subdivision of the species tree and one containing the gene tree (in both files each node is associated {with} an id), the reconciliation in a format compatible with the reconciliation editor SylvX \cite{SylvX} and finally a file containing  a phylogenetic network in eNewick format  \cite{eNewick} obtained by adding transfer edges to the species tree (as indicated by the optimal reconciliation). 

Note that, once the optimal dating is return by \ilpeace, a graph containing all optional solutions can be constructed using the software presented in \cite{scornavaccarepresenting}, an implementation of which is available at \url{http://mbb.univ-montp2.fr/MBB/subsection/downloads}. (The ILP solver returns only a single optimal dating, but modern ILP
solvers such as CPLEX offer access to ``solution pools'' from which alternative optima can be sampled. As with many optimization problems there can be exponentially many optima so generation of them all is challenging. However, \ilpeace does of course compute the \emph{score} of the most parsimonious reconcilation, which is an essential first step towards understanding the structure of this optimal space.)

\subsection{Validation on simulated data}\label{subsec:val}
We validated the  \ilpeace method by comparing its performance  {on simulated data to a brute force approach}. 
Note  that, since the aim of this section is to compare the performance of these two methods, we are not concerned here by retrieving the true reconciliation neither the true \DTL events. The comparsion with a brute force approach is
reasonable because it is unclear how one could develop, for example, an ad-hoc branch and bound algorithm to prune the search space intelligently; this is because the underlying mathematical model is far more complex than more classical combinatorial optimization problems. (Indeed, this complexity is a major reason why we chose to address the problem with ILP.) We also chose not to compare \ilpeace to the superficially similar \textsc{RANGER-DTL} software from \cite{Bansal2012} or the \textsc{JANE} algorithm from \cite{CONOW2010}. \textsc{RANGER-DTL} does not guarantee time-consistency (and does not report whether its optimal solution is time-consistent). \textsc{JANE} uses a slightly different event model to the \DTL model used here and we have observed that \textsc{JANE} sometimes returns strictly less
parsiminious solutions than \ilpeace. It is unclear whether this is due to model differences or generation of suboptimal solutions because \textsc{JANE} is a heuristic and offers no formal guarantees that the solutions it finds are optimal (even within its own model).

The brute force approach consists of three steps. Step~1 is to find all possible orderings for the nodes of the undated species tree. This is done by generating all possible linear extensions (i.e. total orders) of the partial ordering implied by the topology of the species  tree \cite{PruesseRuskey94}. Step~2 is to apply, for each  thus obtained dated version of the species tree, the \DTL reconciliation method \cite{doyon2011efficient,scornavaccarepresenting}. Step~3 is to return an ordering minimizing the reconciliation cost over all possible orderings of the species tree.

Note that, to try all possible orderings for the species tree nodes, and thus to solve the undated \DTL reconciliation problem, we would in principle have to loop through all orderings that extend the partial ordering given by the species tree, not only the total ones. Luckily, it can be proven that looping through all total orderings is enough to find a most parsimonious dating of~$S$. The full proof has been omitted and we give here a short sketch. For each two nodes $s,s'$ of $S$ such that $\theta_S(s)=\theta_S(s')$, the value of~$\theta_S(s')$ can be modified to~$\theta_S(s')-\epsilon$ for some small~$\epsilon>0$. If~$\epsilon$ is chosen small enough -- i.e. such that no other date of $S$ falls in the interval $[\theta_S(s')-\epsilon,\theta_S(s')]$ -- this can be shown not to  {negatively} affect the reconciliation. Doing so for each pair of nodes with the same value $\theta_S(\cdot)$  leads to a date function~$\theta'_S$ for which $\theta'_S(s)\neq \theta'_S(s')$ whenever~$s\neq s'$. Such a date function can easily be turned into a total ordering.

\subsubsection{Gene tree simulation}

To simulate  gene trees along a dated  species tree $S$, we start by associating to every branch of $S$ an \emph{activity} $a$  that represents the overall rate at which \DTL  events occur on this branch, along with specific rates for each individual event type $r_{\D},r_{\T},r_{\L}$, with $a=r_{\D}+r_{\T}+r_{\L}$. 
We then use a birth-and-death process \cite{Kendall48} to simulate each gene tree using a scheme similar to what done in \cite{coev}:  

\begin{enumerate}
\item at the beginning of the process, the first gene of the birth-and-death process is located at the root of $S$;
\item at any time, the time $t_{next}$ of the next potential \DTL event in every existing gene is calculated by simulating an exponential variable with parameter equal to the activity of the branch $(x,y)$ containing that gene. 
Then, if $t_{next}\geq \theta_{S'}(y)$, the next event is determined to be a \cont event if $y$ is a leaf, and an \S event otherwise. If $t_{next} < \theta_{S'}(y)$, the next event is a \DTL event and we rely on the relative rates $r_{\D},r_{\T},r_{\L}$ to determine its type.
\item we repeat this process until we reach the time of the extant species. 
\end{enumerate}
Only gene trees with at least 2 leaves are retained. 

For the simulations, we chose as dated species tree a phylogeny of 37 proteobacteria over a period of 500 million years \cite{david2010rapid}. To be able to apply the brute force approach, we were forced to run the gene tree simulation using  a version of this phylogeny  {restricted} to 17 leaves. 

The duplication, transfer and loss rates were generated for each simulated gene independently and were chosen in accordance with real dataset observations \cite{david2010rapid} using the same scheme as \cite{nguyen2013reconciliation}: the loss rate was randomly chosen in the interval [0.001, 0.0018], where the units are events per gene per million years; the ratio between the ``birth'' rate (sum of the duplication and transfer rates) and the loss rate was randomly chosen in the interval [0.5,1.1]; ultimately the proportion of the duplication rate to the birth rate was randomly chosen in the interval [0.7,1]. We chose to simulate gene trees using a scheme similar to the one used in \cite{coev} rather than the one used  in \cite{nguyen2013reconciliation} to avoid to simulate sequence alignments, computationally costly and not needed in this paper.

We simulated 100 gene trees and we applied both the \ilpeace and the  brute force approach to each pair species tree - gene tree, using the following cost vector: $c_\D=2$, $c_\T=3$ and $c_\L=1$ (and $c_\SL=1$, $c_\TL=4$ and $c_\S=c_\cont=0$).  {We ran the experiments
on a 3.2 GHz Intel Core i3 processor with 8 GB of RAM, using CPLEX 12.5 as the ILP solver.}

\subsubsection{Results}\label{sec:results}
The results of the simulations are shown in Table~\ref{tab}. 

\begin{table}[h!]
\caption{Minimum, average and maximum values for the number of taxa in the gene trees, reconciliation scores, running times of the brute force approach, the no-$\TL$-bound and \ilpeace, and the approximation ratio of the no-$\TL$-bound.  {Running times are in seconds.}\label{tab}}
\begin{center}
\begin{tabular}{l|c|c|c|c|c|c|}
\cline{2-7}
& & & \multicolumn{3}{c|}{running time} & app. ratio no-$\TL$\\
\cline{4-6}
& taxa & score & brute force & no-$\TL$ & \ilpeace &  over the true optimum\\
\hline
\hspace{-3.1mm} \vline \hspace{2.4mm} min. & 2 & 0 & 384 & 2.3 & 6 & 1\\
\hspace{-3.1mm} \vline \hspace{2.4mm} avg. & 11.2 & 12.2 & 748.1 &  5.66 & 18.6 & 1.0007\\
\hspace{-3.1mm} \vline \hspace{2.4mm} max. & 20 & 40 & 1968 & 24.3 & 185.6 & 1.07\\
\hline
\end{tabular}
\end{center}
\label{default}
\end{table}

 {As can be seen from the table, \ilpeace is on average between one and two orders of magnitude faster than the brute force algorithm. Note also that the no-$\TL$-bound, which appears to be significantly easier for CPLEX to compute than
the true optimum, is nevertheless an extremely good approximation of the true optimum.}

%
%

\section{Conclusions}

We have shown that in the absence of complete dating information for the species tree, ILP can be a powerful tool
for computing a most parsimious temporally feasible reconciliation ranging over the space of all possible datings. This is the first algorithm with such properties, since competing approaches either fail to guarantee optimality (i.e. that the solution is most parsimonious) or feasibility (i.e. that the solution is time consistent). For the trees we used (up to 20 taxa) the ILP formulation was typically 10-100 times faster than the obvious brute force algorithm. 
This is significant given that it is far from obvious how to develop an \emph{ad hoc} pruning algorithm to substantially improve upon the brute force approach. The average running time of the ILP formulation was 18 seconds.
For larger trees (say, 30 taxa or more, which is already far out of reach of a naive brute force algorithm) the performance of \ilpeace unfortunately begins to deteriorate. It will still terminate quickly in many cases but worst-case running times start to explode. Nevertheless, we are optimistic that \ilpeace represents an important proof-of-concept, for the following reasons. Firstly, ongoing research into an improved ILP formulation is likely to yield a significant improvement in running times; {the most obvious starting point is to understand why
 $\TL$ events are a bottleneck to fast ILP execution.} {Indeed, in the combinatorial optimization literature it is standard practice to introduce an initial ILP formulation which then in subsequent publications is steadily refined, see e.g. \cite{lancia2009set}}. Secondly, ILP is a very attractive approach because of its inherent flexibility: extra knowledge (such as partial dating information) can easily be incorporated into the
formulation. Moreover, the formulation can easily be extended to allow simultaneous reconciliation
of many gene trees with a given species tree.  As mentioned in the introduction this will potentially enable us to use several gene trees to impose a ``most parsimonious dating'' on a single species tree.



\section*{Competing interests}
  The authors declare that they have no competing interests.

\section*{Author's contributions}
LvI, CS and SK designed the model and solution approach, verified its correctness, designed and conducted the experiments and wrote the paper. LvI implemented the software.

\section*{Acknowledgements}
The authors would like to thank Eric Tannier and Bastien Boussau for fruitful discussions. 


\bibliographystyle{bmc-mathphys}


\bibliography{ILPEACE}      

\end{document}